\documentclass[sigconf,preprint,prologue,table, dvipsnames]{acmart}
\usepackage{xcolor}
\usepackage[utf8]{inputenc} 
\usepackage[T1]{fontenc}    
\usepackage{hyperref}       
\usepackage{url}            
\usepackage{booktabs}       
\usepackage{amsfonts}       
\usepackage{nicefrac}       
\usepackage{microtype}      
\usepackage{verbatim}
\usepackage{graphicx}
\usepackage{amsmath}
\usepackage{algorithm}
\usepackage{color}
\usepackage{epstopdf}
\usepackage[applemac]{inputenx}
\usepackage{algpseudocode}
\usepackage{array}
\usepackage[title]{appendix}
\usepackage{amsthm}
\usepackage{mathtools}
\usepackage{subcaption}
\usepackage{wrapfig}

\newtheorem{theorem}{Theorem}[section]
\newtheorem{corollary}{Corollary}[theorem]
\newtheorem{lemma}[theorem]{Lemma}

\title{Improving precision of A/B experiments using trigger intensity}

\author{Tanmoy Das}
\email{tanmdas@amazon.com}
\affiliation{%
  \institution{Amazon}
  \country{USA}
}
\author{Dohyeon Lee}
\email{dohnlee@amazon.com}
\affiliation{%
  \institution{Amazon}
  \country{USA}
}
\author{Arnab Sinha}
\email{arsinha@amazon.com}
\affiliation{%
  \institution{Amazon}
  \country{USA}
}

\begin{abstract}
In industry, online randomized controlled experiment (a.k.a. A/B experiment) is a standard approach to measure the impact of a causal change. These experiments have small treatment effect to reduce the potential blast radius. As a result, these experiments often lack statistical significance due to low signal-to-noise ratio. A standard approach for improving the precision (or reducing the standard error) focuses only on the \textit{trigger} observations, where the output of the treatment and the control model are different. Although evaluation with full information about trigger observations (\textit{full knowledge}) improves the precision, detecting all such trigger observations is a costly affair. In this paper, we propose a sampling based evaluation method (\textit{partial knowledge}) to reduce this cost. The randomness of sampling introduces bias in the estimated outcome. We theoretically analyze this bias and show that the bias is inversely proportional to the number of observations used for sampling. We also compare the proposed evaluation methods using simulation and empirical data. In simulation, bias in evaluation with partial knowledge effectively reduces to zero when a limited number of observations $(\le 0.1\%)$ are sampled for trigger estimation. In empirical setup, evaluation with partial knowledge reduces the standard error by 36.48\%.

\end{abstract}

\begin{document}

\maketitle

\section{Introduction} 
In industry, online randomized controlled experiment, known as A/B experiment, \cite{power_of_exp, bakshy2014designingdeployingonlinefield, xu_infra_to_culture_ab_test, kaufman2017democratizingonlinecontrolledexperiments, kohavi_practical_guide, diane_overlapping_exp, kohavi_trustworthy, xie_improving_sensitivity} is conducted to evaluate the impact of any causal change. The main challenge in conducting A/B experiment in industrial setup is the low signal-to-noise ratio \cite{xu_infra_to_culture_ab_test, Smith2019} as changes made in the treatment model are incremental in nature that impacts only a few observations. Hence, these experiments have higher standard error (lower precision) and lacks statistical significance. Most of the time, there are missing opportunities to roll out good treatments to production which would improve customer experience and overall revenue.

A general approach \cite{deng_diluted_treatment} to improve precision considers only the \textit{trigger} observations. In a trigger observation, the output of the control and the treatment model are different. Therefore, it is reasonable to assume that the treatment effect is restricted only to these trigger observations.

An evaluation method with full information (\textit{full knowledge}) about trigger observations would be the ideal solution. However, detecting all trigger observations for an experiment is expensive as there can be billions of observations \cite{amzn_traffic_stat_1}, \cite{amzn_traffic_stat_2} per day when conducting an A/B experiment.

In this paper, we propose an alternative solution. Instead of detecting all trigger observations, we use a sampling based approach (\textit{partial knowledge}) where we randomly select a subset of observations and determine their trigger status. Trigger information from these samples are used in evaluation.

Obviously, this sampling based approach introduces bias in the evaluation outcome. But our theoretical analysis shows that this bias reduces linearly as the number of samples increases. Thus, we believe it is a promising approach. The theoretical findings are confirmed by simulated data. 

This new evaluation method is deployed in an online A/B testing platform to test its performance. In the empirical analysis, evaluation with partial knowledge decreases the standard error of the evaluation outcome by 36.48\% without any detectable bias in the estimated treatment effect.
This paper has the following major contributions:\\
\indent \textbf{1)} To the best of our knowledge, this is the first work that introduces the idea of A/B experiment evaluations with sampled trigger observations.
\\\indent \textbf{2)} The performance of the proposed evaluation method is analyzed and compared with other baseline methods theoretically. 
\\\indent \textbf{3)} Our theoretical analysis is further validated by the simulation and empirical data collected from a real A/B experiment platform.

\section{Related Works}
The design and analysis of A/B experiments is a well-studied subject in statistics \cite{box2005statistics, gerber2012field}. Due to its effectiveness in the detection of causal change, A/B experiment is a widely used mechanism in industry when making data-driven decision\cite{power_of_exp, bakshy2014designingdeployingonlinefield, xu_infra_to_culture_ab_test, kaufman2017democratizingonlinecontrolledexperiments, kohavi_practical_guide, diane_overlapping_exp, kohavi_trustworthy, xie_improving_sensitivity}. There are several works \cite{xu_infra_to_culture_ab_test, kaufman2017democratizingonlinecontrolledexperiments, kohavi_practical_guide, diane_overlapping_exp} that discuss the proper guidelines for conducting A/B experiments in industrial settings, which includes method for randomization, experiment design, engineering infrastructure, and choice of metrics. They also list a number of challenges that are difficult to solve. 

One of the vexing problem for any A/B experiment is the issue of low precision, which is widely known as \textit{sensitivity} problem \cite{Smith2019, kohavi_practical_guide, xu_infra_to_culture_ab_test}. Precision is measured as the inverse of the variance of the evaluation results. Larger variance (lower precision) makes it harder to detect any change caused by the treatment model and prevents the launch of a good feature to the production.

There are two typical solutions to resolve the issue of low precision: increasing the number of samples and choose a treatment that has large impact. However, it is not always easy to increase the number of samples as there are multiple parallel experiments running simultaneously. Also, finding a treatment with big effect size is not always possible because most changes are incremental in nature. 

In the literature, there are three approaches proposed to improve precision when conventional solutions fail -- perform trigger analysis, better evaluation methods to reduce variance, and design evaluation metrics with lower variance.

The concept of trigger analysis is proposed in \cite{xu_infra_to_culture_ab_test, kohavi_practical_guide, deng_improve_sensitivity}. These works consider only those products/users, that actually participated in the experiment, during evaluation. 

In \cite{deng_diluted_treatment}, authors look into the more granular trigger data, where they not only consider those products/customers participating in an experiment, but also the number of trigger observations associated with a product/customer. Such granular information provide valuable insights, but the cost of gathering such detailed information becomes prohibitive at scale with millions of experimental participants. In our solution, we use sampling to reduce the cost and analyze the impact of the sampling noise on the evaluation results.
 
Two methods for improving precision are proposed in \cite{deng_improve_sensitivity}, which are also analyzed and extended by \cite{xie_improving_sensitivity}. The first method involves a clustering based on pre-experiment covariates to reduce the between cluster variance, which leads to lower variance for aggregated result. The second method involves inclusion of pre-experiment covariates in the regression for evaluation that can also minimize the variance. Our proposed method is similar to these approaches, and it shows better outcome when implemented in conjunction with these variance reduction techniques.

Developing methods to design new evaluation metrics are discussed in 
\cite{kharitonov_learning_sensitive, deng_data_driven}. These evaluation metrics have lower variance. Our evaluation method is agnostic of evaluation metric. But the problem with any derived metric is the low interpretability, which makes it hard to explain the experiment results.
\vspace{-3mm}
\section{System overview}
In this section, we provide examples of trigger and non-trigger observations. Also, an overview of the process for detecting trigger observations.

Suppose, there is an e-commerce retailer with a large product inventory. There are multiple images associated with a product and on the product description page, these images need to be displayed in a specific order. The list of images is context dependent and can change from one customer visit to another, which can change the ranking too.

Ranking of these images is an important task as showing good images at the top of the website improves customer response, which can be measured by number of customer visits, number of products brought, etc. A customer visit to the product website is denoted as an observation. Customer response for all observations is recorded for evaluation purpose.

A ranking model is used to determine the ranking order for these images. There is already an old version of ranking model deployed in production, which is known as control model. An updated image ranking model is proposed, which is known as the treatment model. 

Although the treatment and the control model are not the same, there are observations where the ranking of images produced by treatment and control models are the same. These are known as \textit{non-trigger} observations. In contrast, a \textit{trigger} observation is the one where the ranking of images produced by treatment and control model are different. Please see Figure \ref{fig:trig_obs_fig} for an example.

Before rolling out the treatment model into the production, we need to make sure that the treatment is better than the control. To compare their performance, we conduct A/B experiment \cite{power_of_exp}. 

There are two possible ways to conduct the A/B experiment in this scenario: product randomization and customer randomization. Let us talk about product-randomized A/B experiments first, and later we discuss the customer randomized experiments.
\\\noindent\textbf{Product-randomization:} For product-randomized experiments, all products are randomly divided into two mutually exclusive groups -- control and treatment. Random division ensures that products assigned to treatment and control groups have similar characteristics. 

When customer visits a product from the control group, images associated with that product are ranked by the control model. Likewise, when customer visits a product from treatment group, treatment model is used to rank images. 

Irrespective of a product's assignment to control/treatment group, we can determine the trigger status of an observation associated with a product. The \textit{counterfactual} ranking output is generated in the backend and compared with actual ranking output that is stored in a database.

For a control product, ranking produced by the control model is used to rank images on the product website. In the backend, the \textit{counterfactual}, which is the ranking generated by the treatment model, is computed. These two rankings are compared to determine the trigger status of the observation.

We can use the same process to determine the trigger status of an observation associated with a treatment product. Please see Figure \ref{fig:trig_obs_fig} for a detailed example. 

The trigger detection can run in online or offline mode. In online mode, both treatment and control ranking outputs are generated simultaneously and compared. In offline mode, counterfactual output is generated later and compared with actual output that is stored in a database.

\noindent\textbf{Customer-randomization:} As mentioned before, it is possible to use customer-randomized experiments in this scenario, where customers are divided into control and treatment groups. The control model is used to rank images when a control customer visits the product description page, while the treatment model is used for treatment customers. Based on this setup, we can redefine the trigger observations for control and treatment customers. 

The choice of product vs customer randomized experiments depends on the business use case. But in this paper, we assume the underlying A/B experiment uses product-randomization. However, the evaluation methods presented here are equally applicable to customer-randomized experiments as well.

It should be noted that when number of observations are in the order of billions per day \cite{amzn_traffic_stat_1}, \cite{amzn_traffic_stat_2}, determining the trigger status for all observations has a huge cost. This cost is attributed to the re-computation of the ranking output for all observations. Specifically, for control observation, we use the control model to compute the ranking output to display images. In this case, we need to re-compute the ranking output using the treatment model again and compare these two outputs to determine the trigger status. Similar re-computation needs to be performed for treatment products. 

\begin{figure}[h]
    \centering
    \includegraphics[width=0.48\textwidth]{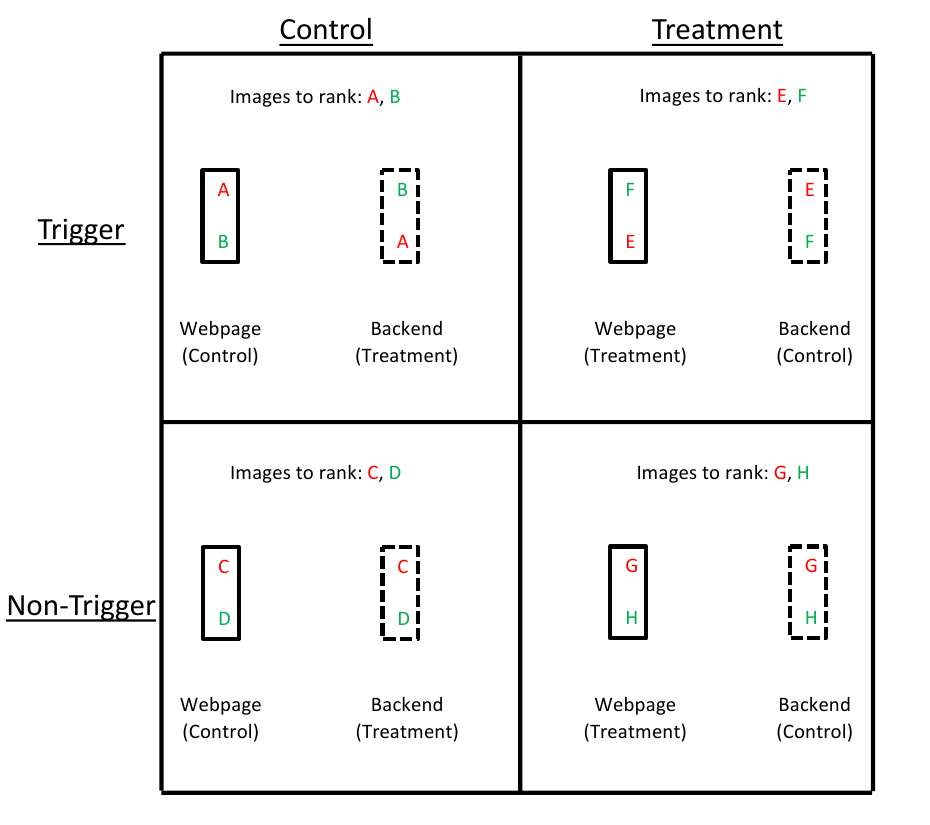}
    \caption{Illustrative example of trigger vs non-trigger observations. Here, we show four observations. Two for control products (left) and two for treatment products (right). Top observations are examples of triggers and the bottom examples are non-triggers. \textit{In top left}, there are two images (\textcolor{red}{A} and \textcolor{green}{B}) to rank. As the product is in control, ranking for product webpage, which is visible to customer, is produced by the control model. \textcolor{red}{A} has the highest rank so it's placed on the top. In the backend, the same images are ranked by the treatment model. As the control and treatment outputs are different. This observation is denoted as trigger. \textit{In bottom left}, the control and treatment model output is the same, so this is a non-trigger observation. Similar analysis is done for treatment product (top right and bottom right). For treatment product, treatment model is used to rank images in webpage and control model is used in backend to determine trigger status.}
    \label{fig:trig_obs_fig}
    \vspace{-5mm}
\end{figure}
\begin{table}[h]
    \centering
    \begin{tabular}{|c|c|}
        \hline
        \textbf{Symbol} & \textbf{Definition} \\
        \hline
        $n_{i}$ & number of observations for $i^{th}$ product\\
        \hline
        $y_{ij}$ & customer response for $i^{th}$ product and\\
         & $j^{th}$ observation\\
        \hline
         $y_{i}$ & average customer response for $i^{th}$ product\\
        \hline
         $r_{ij}$ & trigger status for $i^{th}$ product and\\
         & $j^{th}$ observation\\
        \hline
        $r_{i}$ & average trigger status for $i^{th}$ product\\
        & known as \textit{trigger intensity}\\
        \hline
        $r_i^{\prime}$ & estimated trigger intensity for $i^{th}$ product\\
        \hline
        $\epsilon_i$ & error in estimated trigger intensity \\
        & for $i^{th}$ product\\
        \hline
        $\mathcal{T}_i$ & treatment assignment for $i^{th}$ product\\
        \hline
        $\rho$ & treatment effect\\
        \hline
        $\rho^{\prime}$ & treatment effect when using estimated trigger intensity\\
        \hline
        $\eta$, $\eta^{\prime}$, $\eta_{\alpha}$ & indicates noise\\
        \hline
        $\hat{ }$ & hat symbol indicates estimated parameters\\ 
        \hline
        $\textbf{X}$ & bold symbols indicate matrix and vectors\\ 
        \hline
    \end{tabular}
    \caption{Definition of mathematical symbols.}
    \label{tab:cust_response}
    \vspace{-3mm}
\end{table}

\section{Customer response model}
In this section, let us define the customer response model formally. We use this model for our theoretical analysis.

The treatment model's impact is assumed to affect only trigger observations, since control and treatment model outputs differ only for trigger observations. For non-trigger observations, customer response remains similar between control and treatment products.

A linear function is preferred \cite{deng_improve_sensitivity, xie_improving_sensitivity, athey2018design, Rosenblum2009-ao,} for modeling customer response in A/B experiments due to its robustness and interpretability. Also, random product assignment ensures each treatment item has a similar control counterpart, allowing a linear model to effectively capture any differences in customer response between treatment and control products.

Table \ref{tab:cust_response} shows how the customer response changes for trigger vs non-trigger observations that are associated with control and treatment products. Here, $\mathcal{T}_i$ indicates the treatment assignment status. $\mathcal{T}_i = 1$ means $i^{th}$ product is in treatment. $r_{ij}$ is a binary variable that indicates the trigger status of the $j^{th}$ observation for $i^{th}$ product.  $r_{ij} = 1$ indicates this is a trigger observation.

\begin{table}[h]
    \centering
    \begin{tabular}{|c|c|c|}
        \hline
        & Non-Trigger $(r_{ij} = 0)$ & Trigger $(r_{ij} = 1)$ \\
        \hline
        Control $(\mathcal{T}_i = 0)$ & $\beta_0$ & $\beta_0 + \beta_1$ \\
        \hline
        Treatment $(\mathcal{T}_i = 1)$ & $\beta_0$ & $\beta_0 + \beta_1 + \beta_2$ \\
        \hline
    \end{tabular}
    \caption{ $\beta_0$ is the average customer response for any non-trigger observations. Trigger control observations are different from non-trigger control observations; otherwise they would not be triggered. $\beta_1$ is the additional change in customer response in this case. Thus, customer response for a control trigger observation is $\beta_0 + \beta_1$. Customer response for treatment trigger observations is different from control trigger observations because the treatment displays different output to customer. $\beta_2$ is the additional change in customer response in this case. Hence, customer response for treatment trigger observation is $\beta_0 + \beta_1 + \beta_2$}
    \label{tab:cust_response}
    \vspace{-5mm}
\end{table}
Based on this relationship (Table \ref{tab:cust_response}), we can model the customer response using a linear function as follows. Suppose, $i^{th}$ product has $n_i$ observations. For $j^{th}$ observation associated with $i^{th}$ product, the customer response is $y_{ij}$.
\begin{equation}
\label{eqn:gv_trig}
y_{ij} = \beta_0 + \beta_1 r_{ij} + \beta_2\mathcal{T}_ir_{ij} + \eta_{ij}
\end{equation}
Here, $\eta_{ij}$ is the noise. Since, average customer response for trigger/non-trigger observations does not change, we aggregate observations at the product level. It should be noted that for a customer-randomized experiment, observations would be aggregated at customer level.
\begin{equation}
\label{eqn:dgp}
y_{i} = \beta_0 + \beta_1 r_{i} + \beta_2\mathcal{T}_ir_{i} + \eta_{i}
\end{equation}
Here, $y_i = \frac{1}{n_i}\sum\limits_{j=1}^{n_i} y_{ij}$. $y_i$ is the average customer response for $i^{th}$ product over $n_i$ observations. Similarly, $r_{i} = \frac{1}{n_i}\sum\limits_{j=1}^{n_i} r_{ij}$ and $\eta_{i} = \frac{1}{n_i}\sum\limits_{j=1}^{n_i} \eta_{ij}$. Here, $r_i$ is the average trigger rate across all observations for the $i^{th}$ product. Henceforth, it is known as \textit{trigger intensity}.

If we know values for all $r_{ij}$, we could easily determine the values for $r_i$ and use them to estimate those parameters in Eq \ref{eqn:dgp}. However, detecting all trigger observations $(r_{ij} = 1)$ during an experiment is cost-prohibitive.

An alternative approach for reducing this cost is sampling a smaller subset of observations and perform the re-computation on these observations to estimate the trigger intensity. This estimated trigger intensity ($r_i^{\prime}$) can also be used to evaluate the performance of the treatment model. 

\section{Proposed evaluation methods}
This section presents a theoretical analysis of three evaluation methods: 1) the \textit{baseline} method without trigger information; 2) evaluation with \textit{full knowledge} of trigger intensities; and 3) our proposed approach using \textit{partial knowledge}. Our key contribution demonstrates that in the partial knowledge method, bias in estimated treatment effect and standard error decrease linearly with the number of samples used for trigger intensity estimation.


\subsection{Baseline evaluation: no knowledge of trigger intensity}
\label{sec:baseline_model}
In this section, we present a baseline evaluation method that does not use any information about trigger intensity. We use this model as a baseline to compare the performance of models that utilize the trigger intensity. In particular, this model assumes all units are impacted by the treatment in the same way. 

In this scenario, customer response for all control products are the same. Likewise, customer response for all treatment products are also the same. The difference between the customer response for control and treatment products is caused by the treatment.

\begin{align}
\label{eqn:baseline_eval}
y_{i} = \alpha_0 + \alpha_1 \mathcal{T}_i + \eta_{i}
\end{align}
Here, the customer response for control product is $\alpha_0$ and the additional change in customer response for the treatment product is $\alpha_1$. Hence, $\alpha_1$ represents the average difference in customer response in between the control and the treatment products. It is popularly known as \textit{average treatment effect (ATE)}. 

$\alpha_1$ is estimated using the Ordinary Least Square (OLS) method. This estimation is unbiased when the noise term $(\eta_{i})$ is i.i.d and zero mean Gaussian.
\begin{theorem}
\label{th:baseline_eval}
    With no knowledge of trigger intensity\\
    a) The estimated average treatment effect is
    \begin{align}
    \hat{\alpha_1} = E[y_i | \mathcal{T}_i = 1] - E[y_i | \mathcal{T}_i = 0]
    \end{align}
    b) Suppose, the variance of the residual is $\sigma^2(\hat{\eta_{\alpha}})$. The variance of the estimated treatment effect is
    \begin{align}
    \label{eqn:baseline_model_var}
    \sigma^2 \left( \hat{\alpha_1} \right) = 4\sigma^2(\hat{\eta_{\alpha}})
    \end{align}
\end{theorem}
\begin{proof}
    Proof is in Appendix.
\end{proof}
The problem with the baseline method is that it assumes all treatment products are equally impacted by an experiment. However, it is not a realistic scenario. 

As an instance, there are popular products \cite{item_popularity}, \cite{product_success_popularity} where number of customer visits are much higher than other products. It is reasonable to assume these products are going to have many trigger observations, thus more treatment impact.

In the next section, we present an evaluation method that can overcome this pitfall.

\subsection{Evaluation: full knowledge of trigger intensity}
\label{sec:full_know}
In most practical scenarios, products have varying number of observations. As a result, the number of trigger observations as well as the trigger intensity are also going to vary from one product to another. 

Trigger intensity is a measure of the treatment intensity for a product, and it is well known \cite{kelley_treatment_intensity, yoder_treatment_intensity} that by accounting for changes in treatment intensity, we can reduce the variance of the estimated treatment effect.

Hence, it is beneficial to use the trigger intensity information in the evaluation method so that we can precisely estimate the treatment impact. We use the same linear model as defined in Eqn \ref{eqn:dgp}. Based on this model, we can determine the ATE defined as $\rho$.

As products are randomly divided into treatment and control, we assume trigger intensity $(r_i)$ is independent of treatment assignment $(\mathcal{T}_i)$. 
\begin{theorem}
\label{th:ate_trig}
Average treatment effect $(\rho)$ is $\beta_2 E[r_i]$
\end{theorem}
\begin{proof}
Suppose, the treatment impact for the $i^{th}$ product is $\rho_i$. The average treatment effect for all products with trigger intensity $r_i$ is
\begin{align}
\label{eqn:avg_treat_effect_1}
E[\rho_i | r_i] &= E[y_{i} | \mathcal{T}_i = 1, r_i] - E[y_{i} | \mathcal{T}_i = 0, r_i]\\
&= (\beta_0 + \beta_1 r_i + \beta_2 r_i) - (\beta_0 + \beta_1 r_i)\\
&= \beta_2 r_i 
\end{align}
The average treatment effect across all products is
\begin{align}
\label{eqn:avg_treat_effect}
\rho = E[\rho_i] &= E[E[\rho_i | r_i]] = E[\beta_2 r_i] = \beta_2 E[r_i]
\end{align}
\end{proof}
We estimate the value of $\beta_2$ using OLS.
\begin{theorem}
\label{th:eval_all_obs_equal_and_full_trig_1}
With full knowledge of trigger intensity,\\
a) the estimated value of parameter $\beta_2$ is
\begin{align}
\hat{\beta}_2 &= \frac{1}{E[r_i^2]} \left[ E[y_i r_i | \mathcal{T}_i = 1] - E[y_i r_i | \mathcal{T}_i = 0] \right] 
\end{align}
b) suppose, the variance of the residuals is $\sigma^2\left( \hat{\eta} \right)$, the variance of estimated $\beta_2$ is
\begin{align}
\label{eqn:baseline_ate_var}
\sigma^2 \left( \hat{\beta}_2 \right) &= \frac{4}{E[r_i^2]} \sigma^2\left( \hat{\eta} \right)
\end{align}
\end{theorem}
\begin{proof}
    Proof is in Appendix.
\end{proof}
This estimated $\hat{\beta}_2$ is unbiased as long as the noise term in Eq \ref{eqn:dgp} is i.i.d and zero mean Gaussian. We can compute the value of $E[r_i]$ as we have the full knowledge of the product trigger intensity. Based on this information, we can estimate the ATE and its variance.

\begin{corollary}
\label{th:eval_all_obs_equal_and_full_trig_2}
With full knowledge of product trigger intensity,\\
a) the estimated ATE $(\hat{\rho})$
    \begin{align}
    \hat{\rho} &= \frac{E[r_i]}{E[r_i^2]} \left[ E[y_i r_i | \mathcal{T}_i = 1] - E[y_i r_i | \mathcal{T}_i = 0] \right]
    \end{align}
b) the variance of the estimated ATE $(\sigma^2 \left( \hat{\rho} \right))$ is
    \begin{align}
    \label{eqn:ate_var_full_know}
    \sigma^2 \left( \hat{\rho} \right) &= \sigma^2 \left( \hat{\beta}_2 \right) E[r_i]^2 \\
    &= \frac{4 E[r_i]^2}{E[r_i^2]} \sigma^2\left( \hat{\eta} \right)
    \end{align}
\end{corollary}
\begin{proof}
    The proof follows from Theorem \ref{th:ate_trig} and \ref{th:eval_all_obs_equal_and_full_trig_1}.
\end{proof}

The model with full knowledge of product trigger intensity has a smaller variance for the estimated ATE in comparison to the baseline model (Section \ref{sec:baseline_model}) because the proposed model takes into account the product level variations in the treatment impact. 

The variance of residual for baseline method $(\sigma^2(\hat{\eta_{\alpha}}))$ is larger than the variance of residual for the proposed method with trigger intensity, $(\sigma^2\left( \hat{\eta} \right))$ as the proposed method uses trigger intensity as a covariate. 

Suppose, $\sigma^2(\hat{\eta_{\alpha}}) = h \sigma^2\left( \hat{\eta} \right), h \ge 1$. The ratio of variance for ATE in the proposed model and the baseline model is $\frac{\sigma^2 \left( \hat{\rho} \right)}{\sigma^2 \left( \hat{\alpha_1} \right)} = \frac{E[r_i]^2}{h E[r_i^2]}$ from Eq \ref{eqn:baseline_ate_var} and \ref{eqn:ate_var_full_know}.

When $\sigma^2(r_i) = E[r_i^2] -  E[r_i]^2 > 0$, $\frac{E[r_i]^2}{E[r_i^2]} < 1$. Thus, the variance for the ATE is always smaller for the proposed model when $\sigma^2(r_i) > 0$. The ratio $\frac{E[r_i]^2}{E[r_i^2]}$ is minimized when the $\sigma^2(r_i)$ is maximized. For any given $E[r_i] = \frac{n}{N}, 1 \le n \le N$, we can show that $\sigma^2(r_i)$ is maximized when exactly $n$ products have $r_i = 1$ and rest of them have $r_i = 0$, which means $\frac{E[r_i]^2}{E[r_i^2]} = \frac{n}{N}$. Thus, for any given $E[r_i] = \frac{n}{N}$, the maximum reduction in variance is $\frac{n}{hN}$. This reduction becomes larger as $n$ reduces.

When $r_i$ is the same for all products, $\sigma^2(r_i) = 0$ and  $\frac{E[r_i]^2}{E[r_i^2]} = 1$. In this case, the reduction in variance is $\frac{1}{h}$. A special case happens, when $r_i = 1$ for all products. In particular, all observations are triggered. In this case, Eq \ref{eqn:baseline_eval} and Eq \ref{eqn:dgp} are equivalent. The variance of residuals for both methods is the same and $h = 1$. For this special case, the variance of ATE for both methods is also the same.

In the other extreme, when $r_i = 0$ for all products, there is no trigger observations. Hence, there is no treatment effect and ATE is zero by definition.

\subsection{Evaluation: partial knowledge of trigger intensity}
As mentioned before, determining the trigger status $(r_{ij})$ of all observations is expensive, hence we cannot accurately determine the value of product trigger intensity $(r_i)$. But we can estimate the value of $r_i$ using an inexpensive method that can be noisy. 

Assume, the estimated value is $r_i^{\prime}$. Suppose, $r_i^{\prime} = r_i + \epsilon_i$. In this case, the estimation error ($\epsilon_i$) in the product trigger intensity becomes part of the noise as $\epsilon_i$ is unknown \cite{pischke_lectures}, \cite{hyslop_imbens_measurement_error}. As products are randomly divided into control and treatment groups, it is safe to assume that $r_i^{\prime} \bot T_i$ and $\epsilon_i \bot T_i$. Further, we also assume that the error $(\epsilon_i)$ and the true product trigger intensity $(r_i)$ is uncorrelated $(E[\epsilon_i r_i] = 0)$.

Let us re-write Eq \ref{eqn:dgp} to illustrate the impact of $\epsilon_i$. 
\begin{align}
y_{i} &= \beta_0 + \beta_1 [r_i^{\prime} + \epsilon_i] + \beta_2 \mathcal{T}_i [r_i^{\prime} + \epsilon_i] + \eta_i \\
&= \beta_0 + \beta_1 r_i^{\prime} + \beta_2 \mathcal{T}_i r_i^{\prime} +  \eta_i^{\prime}
\end{align}
The impact of $\epsilon_i$ on the outcome ($y_i$) cannot be determined, as the value of $\epsilon_i$ is unknown. The error $\epsilon_i$ contributes to noise $\eta_i^{\prime} = [\beta_1 + \beta_2 \mathcal{T}_i] \epsilon_i  + \eta_i$.

\begin{theorem}
\label{th:ate_part_trig}
     With partial knowledge of product trigger intensity,\\ 
     a) the estimated value of parameter $\beta_2$ is
     \begin{align}
    \hat{\beta_2^{\prime}} &= \frac{1}{E[(r_i^{\prime})^2]} \left[ E[y_i r_i^{\prime} | \mathcal{T}_i = 1] - E[y_i r_i^{\prime} | \mathcal{T}_i = 0] \right] 
    \end{align}
    b) the bias in the estimated value is
    \begin{align}
    \label{eqn:beta_2_prime_and_beta_2}
    E[\hat{\beta_2^{\prime}}] &= \beta_2 \left[ 1 - \frac{ E[\epsilon_i r_i^{\prime}]}{E[(r_i^{\prime})^2]} \right]
    \end{align}
    c) suppose, $\sigma^2 (\hat{\eta^{\prime}})$ is the residual variance and $\sigma^2 (\eta)$ is the noise variance.  When number of products ($N$) is large
    \begin{align}
        \sigma^2(\hat{\eta^{\prime}}) \approx \sigma^2(\eta) + \left[ \left( \beta_1 + \frac{\beta_2}{2} \right)^2 + \frac{\beta_2^2}{4} \right] E[\epsilon_i^2]
    \end{align}
    d) the variance of the estimated $\beta_2$ is
    \begin{align}
        \sigma^2 \left( \hat{\beta_2^{\prime}} \right) = \frac{4}{E[r_i^{{\prime}^2}]}  \sigma^2 (\hat{\eta^{\prime}})
    \end{align}
\end{theorem}
\begin{proof}
    Proof is in Appendix
\end{proof}
As $E[r_i \epsilon_i] = 0$, we can re-write $E[\epsilon_i r_i^{\prime}] = E[\epsilon_i^2] > 0$. It means there is a downward bias in the estimated $\hat{\beta_2^{\prime}}$. The residuals $(\hat{\eta})$ are impacted by the noise $(\epsilon_i)$ in the estimated trigger intensity and the variance of residuals $(\sigma^2 (\hat{\eta^{\prime}}))$ is greater than the variance of the noise $(\sigma^2 (\eta))$ in observation. Using this $\hat{\beta_2^{\prime}}$, we can estimate the ATE with partial knowledge of trigger intensity defined as $\hat{\rho^{\prime}}$. We also estimate the bias and variance of $\hat{\rho^{\prime}}$.
\begin{corollary}
    With partial knowledge of product trigger intensity,\\
    a) the estimated ATE 
    \begin{align}
    \hat{\rho^{\prime}} = \hat{\beta_2^{\prime}} E[r_i^{\prime}]
    \end{align}
    b) the bias in the estimated ATE is
    \begin{align}
    E[\hat{\rho^{\prime}}] = \rho + \beta_2 \left( E[\epsilon_i] -  \frac{ E[\epsilon_i r_i^{\prime}]E[r_i^{\prime}]}{E[{r_i^{\prime}}^2]} \right)
    \end{align}
    c) the variance of the estimated ATE is
    \begin{align}
        \sigma^2 \left( \hat{\rho^{\prime}} \right) = \frac{4 E[r_i^{\prime}]^2}{E[r_i^{{\prime}^2}]}  \sigma^2 (\hat{\eta^{\prime}})
    \end{align}
\end{corollary}
\begin{proof}
    Proof follows from Theorem \ref{th:ate_part_trig} and Theorem \ref{th:ate_trig}. 
\end{proof}
As the estimated $\hat{\beta_2^{\prime}}$ has bias, it also leads to bias in the estimated ATE $(\hat{\rho^{\prime}})$. Similarly, there is a bias in the estimated variance too. In the next section, we discuss a trigger intensity estimation method and present a more detailed analysis for the bias in the estimated ATE and its variance.

\section{Trigger intensity estimation method}
In this section, we present a possible solution for computing the product trigger intensity
and analyze its impact on the estimation bias of the ATE $(\hat{\rho}^{\prime})$ with the partial knowledge of product trigger intensity.
\subsection{Independent sampling to estimate the trigger intensity}
The $i^{th}$ product has $n_i$ observations. Suppose, we randomly sample $m, m < n_i$ out of these $n_i$ observations and compute outputs using both the treatment and the control models. These outputs are compared to determine if there is any change in the model output. As defined before, $r_{ij} = 1$ indicates that there is a difference in the output of the treatment and the control models. The estimated trigger intensity for $i^{th}$ product is $r_i^{\prime} = \frac{\sum\limits_{j=1}^{m} r_{ij}}{m}$.

The estimation error is $\epsilon_i = r_i^{\prime} - r_i$. The mean of the estimation error is zero $(E[\epsilon_i | r_i] = 0)$ and the variance is $\sigma^2 \left( \epsilon_i | r_i \right) = E[\epsilon_i^2 | r_i] = \frac{r_i (1 - r_i)}{m}$.
We can determine the expected value and the variance of the $\epsilon_i$.
\begin{lemma}
    \label{lemma:ub_ate_part_trig_1}
    The $E[\epsilon_i] = 0$ and $\sigma^2(\epsilon_i) = E[\epsilon_i^2] = \frac{E[r_i] - E[r_i^2]}{m_i}$
\end{lemma}
\begin{proof}
    The proof follows from the iterated law of expectation.
\end{proof}
With the help of this knowledge, it is possible to compute an upper bound on the estimation bias for ATE as defined in Theorem \ref{th:ate_part_trig}. 
\begin{theorem}
    \label{th:ub_ate_part_trig}
    If $m$ (where $m \le n_i, \forall i$ and $m > 1$) observations for all products are examined to estimate the product trigger intensity and $\sigma^2 \left( r_i \right) > 0$,\\
    a) there is a downward bias in the estimated ATE $(\hat{\rho^{\prime}})$ and the bias is upper bounded as follows 
    \begin{align}
        \rho - E[\hat{\rho}^{\prime}] < \frac{\beta_2}{m - 1}
    \end{align}
    b) the variance of the estimated ATE $(\sigma^2 \left( \hat{\rho^{\prime}} \right))$ is larger than the variance of the estimated ATE $(\sigma^2 \left( \hat{\rho} \right))$ with full knowledge of trigger intensity. The difference has an upper bound as follows 
    \begin{align}
        \sigma^2 \left( \hat{\rho^{\prime}} \right) - \sigma^2 \left( \hat{\rho} \right) < \frac{1}{m} \left[ \left( \beta_1 + \frac{\beta_2}{2} \right)^2 + \frac{\beta_2^2}{4} \right]
    \end{align}
\end{theorem}
\begin{proof}
    Proof is in Appendix
\end{proof}
There is a downward bias in the estimated ATE when we use the trigger intensity computed from independent sampling. On the other hand, it leads to an upward bias in the variance of the ATE. But both of these biases are inversely proportional to the number of observations $(m)$. 

\begin{figure}
    \begin{subfigure}{0.45\textwidth}
        \centering
            \captionsetup{width=.95\linewidth}
            \includegraphics[width=0.95\textwidth]{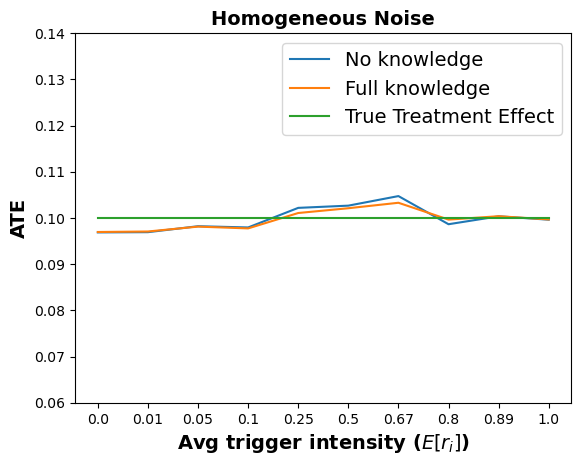}
            \caption{Both evaluation methods are unbiased when noise is homogeneous.}
    \end{subfigure}
    \begin{subfigure}{0.45\textwidth}
        \centering
            \captionsetup{width=.95\linewidth}
            \includegraphics[width=0.95\textwidth]{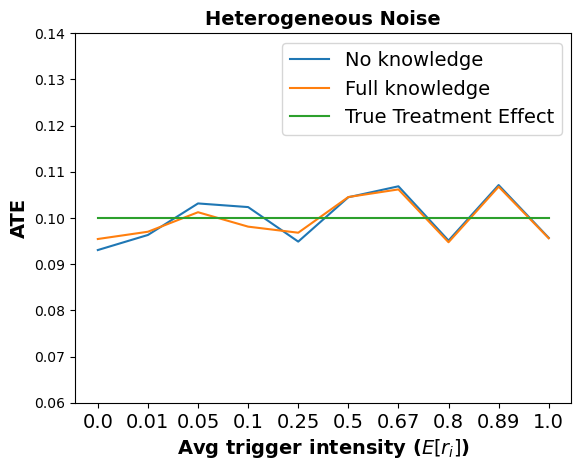}
            \caption{Both evaluation methods are unbiased when noise is heterogeneous.}
    \end{subfigure}
\caption{Evaluations with no knowledge and full knowledge of trigger intensity are unbiased irrespective of the noise characteristics.}
\label{fig:no_know_vs_full_know_ate}
\end{figure}

\begin{figure}
\begin{subfigure}{0.45\textwidth}
        \centering
            \captionsetup{width=.95\linewidth}
            \includegraphics[width=0.95\textwidth]{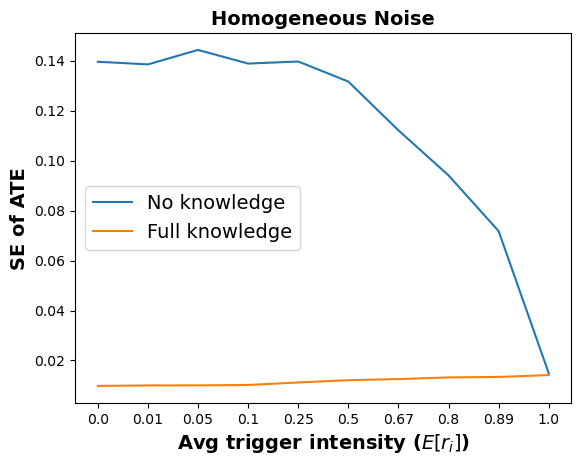}
            \caption{The estimated ATE with full knowledge of triggers is more precise, as indicated by the smaller SE under homogeneous noise.}
    \end{subfigure}
    \begin{subfigure}{0.45\textwidth}
        \centering
            \captionsetup{width=.95\linewidth}
            \includegraphics[width=0.95\textwidth]{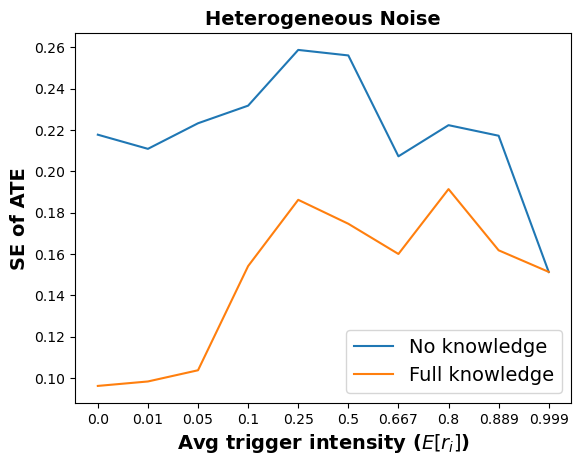}
            \caption{The estimated ATE with full knowledge of triggers is more precise, as indicated by the smaller SE under heterogeneous noise.}
    \end{subfigure}
    \caption{Evaluation with full knowledge of trigger intensity is more precise in comparison to evaluation with no knowledge of trigger intensity irrespective of the noise characteristics.}
    \label{fig:no_know_vs_full_know_ate_se}
\end{figure}

\begin{figure}
    \begin{subfigure}{0.44\textwidth}
	\centering
            \captionsetup{width=.95\linewidth}
		\includegraphics[width=0.95\textwidth]{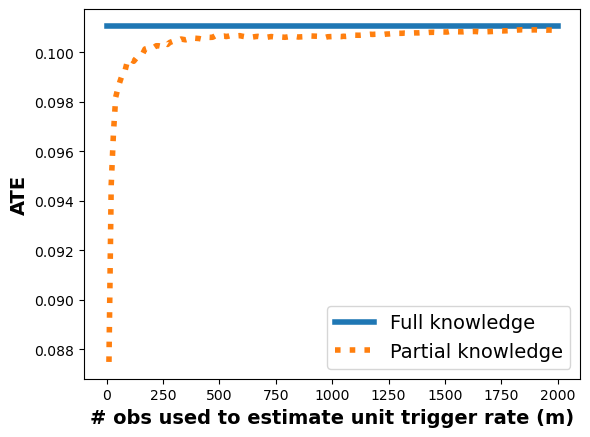}
	\caption{There is a negative bias in the estimated ATE with partial knowledge of trigger intensity. This bias is inversely proportional to the number of observations used to estimate trigger intensity increases.} 
	\label{fig:estimate_independent_prior}
    \end{subfigure}
    \begin{subfigure}{0.44\textwidth}
	\centering
            \captionsetup{width=.95\linewidth}
		\includegraphics[width=0.95\textwidth]{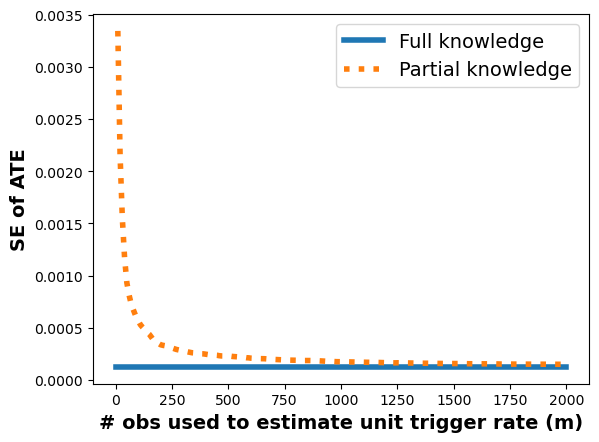}
	\caption{There is a positive bias in the variance of ATE with partial knowledge of trigger intensity. This bias is inversely proportional to the number of observations used to estimate trigger intensity increases.}
	\label{fig:estimate_hierarchical_prior}
    \end{subfigure}
\caption{Comparison of estimated values with full knowledge of trigger intensity and partial knowledge of trigger intensity.}
\label{fig:full_know_vs_part_know}
\end{figure}

\section{Simulation study}
We use simulation to validate the theoretical analysis. Simulation analysis provides further evidence for the following two claims: 1) the standard error of evaluation with full knowledge is smaller than the standard error of the baseline method; and 2) the bias in the evaluation with partial knowledge reduces as the number of observations for sampling increases.

We use 2000 experiment products, which are randomly divided into treatment and control groups. The number of observations for a product can be as high as 1 million. Here, only products that are part of the experiments are considered for analysis, as suggested by \cite{xu_infra_to_culture_ab_test, kohavi_practical_guide, deng_improve_sensitivity}. 

We also vary the noise characteristics (homogeneous vs heterogeneous) to check the robustness of the theoretical results. Homogeneous noise assumes measurement associated with all products have the same noise variance, i.e., noise is i.i.d. For the heterogeneous case, noise is independent, but its variance can change based on products.

\subsection{Full knowledge vs baseline}
Comparing baseline evaluation with full knowledge reveals both methods estimate ATE without bias (Figure \ref{fig:no_know_vs_full_know_ate}). The difference between the estimated ATEs from these two methods reduces as the trigger intensity increases $(> 0.8)$. These methods converge as trigger intensity increases above 0.8. With higher trigger rates, more observations are impacted, reducing the noise in baseline estimates.

Evaluation with full knowledge achieves lower standard error (Figure \ref{fig:no_know_vs_full_know_ate_se}), which increases precision and reduces Type II errors. This improved ability to detect small treatment effects is particularly valuable in industry settings, where subtle improvements in customer experience are common but difficult to validate statistically.

\subsection{Full knowledge vs partial knowledge}
We also evaluate the performance of the evaluation method when there is only partial knowledge of trigger intensity is available. Here, the trigger intensity of a product is estimated by inspecting a sample of observations related to that product. The error in the estimated product trigger intensity ($r_i^{\prime}$) depends on the number of observations ($m$) inspected per product. 

Figure \ref{fig:full_know_vs_part_know} shows the simulated results. As $m$ increases, the error in the estimated trigger intensity reduces, which generates more accurate (lower bias) and precise (lower standard error) estimate of ATE. In this simulation setup, bias in the estimated ATE from evaluation with partial knowledge becomes very small when $m$ is larger than 20 samples. Same thing happens for variance too. 

It should be mentioned that this threshold for the required number of samples will depend on the inherent noise variance and the true treatment effect.

\section{Comparison using real A/B experiments data}
In this section, we present the data from a real A/B experiment platform, where the partial knowledge evaluation method is deployed. In this analysis, a reasonable number of samples are used to compute the trigger intensity. Still, evaluation with partial knowledge reduces the standard error without any observable bias in the estimated ATE.

This A/B experiment platform is used regularly by an e-commerce retailer to measure the causal impact. Here, the treatment model may change the model parameters, update the list of input features, change the engineering infrastructure, etc. 

This platform does not provide the full trigger intensity information, as the cost for determining the trigger status for all observations is huge (in the order of tens of millions of dollars). Thus, we estimate the trigger intensity for a product by sampling a subset of observations and computing their trigger status. The number of observations $(m)$ used for trigger intensity estimation can vary from 10 to 70 samples.

We compare the performance of the evaluation with partial knowledge of trigger intensity and the baseline method. Both evaluation methods only consider those products that participate in an experiment, as suggested by \cite{xu_infra_to_culture_ab_test, kohavi_practical_guide, deng_improve_sensitivity}.

The data is collected over five months. There are 37 experiments and 92 treatments. Some experiments have more than one treatment. The maximum number of treatments for an experiment is six. Some of these experiments also ran in parallel. 
\begin{table}[h]
    \centering
    \begin{tabular}{|c|c|c|}
        \hline
        & Baseline & Partial knowledge \\
        \hline
        Avg standard Error & 0.1781 & 0.11305\\
        \hline
        Avg absolute t-val & 0.53203 & 0.98008\\
        \hline
        Number of statistically & &\\
        significant (90\%) treatments & 27 & 39\\
        \hline
        Number of statistically & &\\
        significant (95\%) treatments & 21 & 31\\
        \hline
    \end{tabular}
    \caption{Comparison of t-value, standard error, and statistically significant treatments for baseline and partial knowledge evaluation methods from an online A/B experiment platform. The t-vale and standard error are aggregated over multiple treatments.}
    \label{tab:cust_response}
    \vspace{-8mm}
\end{table}
\subsection{Reduction of standard error}
Overall, the standard error has reduced by 36.48\%. Because of this, the number of statistically significant (95\%) treatment has increased by 44.44\%. This is in line with the theoretical analysis. First, evaluation with perfect knowledge of trigger intensity has much smaller standard error in comparison to baseline method. Second, the bias in the standard error (\ref{th:ate_part_trig}) from evaluation with partial knowledge reduces as the number of observations used for trigger intensity estimation.

Most experiments use enough samples for trigger intensity estimation. As a result, the estimated standard error from evaluation with partial knowledge is very close to the evaluation with perfect knowledge, which is smaller in comparison to baseline method.

We also perform a paired t-test to check if the change in standard error is statistically significant. The p-value of the test is very close to zero, indicating the reduction in standard error is statistically significant. It is another evidence that shows the evaluation with partial knowledge is an effective method for reducing standard error and improving precision.

\subsection{No bias in estimated ATE}
The estimated ATE from evaluation with partial knowledge can be biased. The source of this bias is the estimation error in the trigger intensity. But this bias vanishes when we use enough samples for trigger intensity estimation.

We compare the estimated ATE values from two evaluation methods (baseline vs evaluation with partial knowledge) to check if any observable bias exists. The estimated ATE for evaluation with partial knowledge is slightly higher. 

However, this change is not statistically significant. We perform a paired t-test to determine the statistical significance of this difference. The p-value is 0.4. It means we can safely reject the hypothesis that estimated ATE from evaluation with partial knowledge is different from the baseline method.

We also count the number of times the confidence intervals of these two evaluation methods overlap. There are 92 treatments. For 91 treatments,  confidence intervals from two evaluation methods overlap each other. This is another indication that the difference in estimated ATE values are not statistically significant. 

In other words, there are enough samples used for trigger intensity estimation and there is no observable bias in the estimated ATE from evaluation with partial knowledge.
\begin{table}[h]
    \centering
    \begin{tabular}{|c|c|c|}
        \hline
         & Number of treatments \\
        \hline
        Total treatments & 92\\
        \hline
        95\% confidence intervals overlap & 91\\
        \hline
        90\% confidence intervals overlap & 91\\
        \hline
        Estimated ATEs have & \\
        the same sign & 70\\
        \hline
    \end{tabular}
    \caption{Comparison of estimated ATE for baseline and partial knowledge evaluation methods.}
    \label{tab:cust_response}
    \vspace{-5mm}
\end{table}

\section{Conclusion}
A/B experiments in industry typically have small treatment effects to minimize risks, often resulting in statistically insignificant results due to low signal-to-noise ratios. While focusing on trigger observations (where treatment and control outputs differ) can improve precision, identifying all triggers is resource-intensive. We propose a sampling-based evaluation method that reduces costs while maintaining effectiveness. Our theoretical analysis shows that sampling bias decreases inversely with sample size. Through simulations, we demonstrate that sampling just 0.1\% of observations effectively eliminates bias, while empirical testing shows a 36.48\% reduction in standard error.

\bibliographystyle{unsrt}
\bibliography{ref}

\appendix
\section{Proof of Theorem \ref{th:baseline_eval}}
We derive parameters of the baseline model. We use the matrix notation to represent the input parameters $\mathbf{X}$, which is a matrix of size $N \times 2$. $N$ is the number of observations and there are two parameters in the baseline model: a constant term and the treatment indicator. The treatment indicator is represented by a column vector $\mathbf{t}$  and the constant term is represented by $\mathbf{1}$. Both of them have the size of $N \times 1$. The observed outcome ($\mathbf{y}$) is a column vector of size $N \times 1$. The estimated parameters are denoted by vector $\boldsymbol{\hat{\alpha}} = \left[ \begin{matrix} \hat{\alpha}_0 & \hat{\alpha}_1\end{matrix} \right]$. The variance of the residuals is $\sigma^2(\hat{\eta_{\alpha}})$.

\subsection{Proof of Theorem \ref{th:baseline_eval}: part a}
\begin{align}
&\mathbf{X} = \left[ \begin{matrix}
\mathbf{1} & \mathbf{t}
\end{matrix} \right] \\
&\frac{1}{N}\mathbf{X}^{T} \mathbf{X} = \left[ \begin{matrix}
1 & E[\mathcal{T}_i] \\
E[\mathcal{T}_i] & E[\mathcal{T}_i^2]
\end{matrix} \right] = \left[ \begin{matrix}
1 & \frac{1}{2} \\
\frac{1}{2} & \frac{1}{2}
\end{matrix} \right]\\
&\left( \frac{1}{N}\mathbf{X}^{T} \mathbf{X} \right)^{-1} = 4 \left[ \begin{matrix}
\frac{1}{2} & -\frac{1}{2}\\
-\frac{1}{2} & 1
\end{matrix} \right]\\
&\frac{1}{N} \mathbf{X}^{T} \mathbf{y} = \left[ \begin{matrix}
E[y_i]\\
E[\mathcal{T}_iy_i]
\end{matrix} \right] = \left[ \begin{matrix}
\frac{1}{2} \left( E[y_i | \mathcal{T}_i = 1] + E[y_i | \mathcal{T}_i = 0] \right)\\
\frac{1}{2} E[y_i | \mathcal{T}_i = 1]
\end{matrix} \right]\\
&\boldsymbol{\hat{\alpha}} = \left( \frac{1}{N}\mathbf{X}^{T} \mathbf{X} \right)^{-1} \frac{1}{N} \mathbf{X}^{T} \mathbf{y} = \left[ \begin{matrix}
E[y_i | \mathcal{T}_i = 0]\\
E[y_i | \mathcal{T}_i = 1] - E[y_i | \mathcal{T}_i = 0]
\end{matrix} \right]\\
\end{align}
Hence, $\hat{\alpha}_1 = E[y_i | \mathcal{T}_i = 1] - E[y_i | \mathcal{T}_i = 0]$

\subsection{Proof of Theorem \ref{th:baseline_eval}: part b}
\begin{align}
&\sigma^2(\boldsymbol{\hat{\alpha}}) = \text{diag} \left( \frac{1}{N}\mathbf{X}^{T} \mathbf{X} \right)^{-1} \sigma^2(\hat{\eta_{\alpha}}) =  \left[ \begin{matrix}
2 \sigma^2(\hat{\eta_{\alpha}}) & 4 \sigma^2(\hat{\eta_{\alpha}})
\end{matrix} \right]
\end{align}
Hence, $\sigma^2(\hat{\alpha}_1) = 4 \sigma^2(\hat{\eta_{\alpha}}]$

\section{Proof of Theorem \ref{th:eval_all_obs_equal_and_full_trig_1}}
Trigger intensity is represented by a column vector $\mathbf{r}$ of size $N \times 1$. $\odot$ is a symbol for element wise multiplication. Assume, $E[r_i] = m$, $E[r_i^2] = s$, and $\sigma^2(r_i) = v$. The estimated parameters are $\hat{\beta}$.
\begin{align}
&\mathbf{X} = \left[ \begin{matrix}
\mathbf{1} & \mathbf{r} & \mathbf{r} \odot \mathbf{t}
\end{matrix} \right] \\
&\frac{1}{N}\mathbf{X}^{T} \mathbf{X} = \left[ \begin{matrix}
1 & E[r_i] & E[\mathcal{T}_ir_i]\\
E[r_i] & E[r_i^2] & E[\mathcal{T}_ir_i^2]\\
E[\mathcal{T}_ir_i] & E[\mathcal{T}_ir_i^2] & E[\mathcal{T}_i^2r_i^2]
\end{matrix} \right] = \left[ \begin{matrix}
1 & m & \frac{m}{2}\\
m & s & \frac{s}{2}\\
\frac{m}{2} & \frac{s}{2} & \frac{s}{2}
\end{matrix} \right]\\
&\left( \frac{1}{N}\mathbf{X}^{T} \mathbf{X} \right)^{-1} = \left[ \begin{matrix}
\frac{s}{v} & -\frac{m}{v} & 0\\
-\frac{m}{v} & \frac{1}{v} + \frac{1}{s} & -\frac{2}{s}\\
0 & -\frac{2}{s} & \frac{4}{s}
\end{matrix} \right]\\
&\frac{1}{N} \mathbf{X}^{T} \mathbf{y} = \left[ \begin{matrix}
E[y_i]\\
E[r_iy_i]\\
E[\mathcal{T}_ir_iy_i]
\end{matrix} \right] = \left[ \begin{matrix}
\frac{1}{2} \left( E[y_i | \mathcal{T}_i = 1] + E[y_i | \mathcal{T}_i = 0] \right)\\
\frac{1}{2} \left( E[r_iy_i | \mathcal{T}_i = 1] + E[r_iy_i | \mathcal{T}_i = 0] \right)\\
\frac{1}{2}E[r_iy_i | \mathcal{T}_i = 1]
\end{matrix} \right]\\
&\boldsymbol{\hat{\beta}} = \left( \frac{1}{N}\mathbf{X}^{T} \mathbf{X} \right)^{-1} \frac{1}{N} \mathbf{X}^{T} \mathbf{y} \\
&= \left[ \begin{matrix}
\frac{s}{v} E[y_i] - \frac{m}{v}E[r_iy_i]\\
-\frac{m}{v}E[y_i] + \frac{1}{v} E[r_iy_i] - \frac{1}{2s} \left( E[r_iy_i | \mathcal{T}_i = 1] - E[r_iy_i | \mathcal{T}_i = 0] \right)\\
\frac{1}{s} \left( E[r_iy_i | \mathcal{T}_i = 1] - E[r_iy_i | \mathcal{T}_i = 0] \right)
\end{matrix} \right]\\
&\sigma^2(\boldsymbol{\hat{\beta}}) = \text{diag} \left( \frac{1}{N}\mathbf{X}^{T} \mathbf{X} \right)^{-1} \sigma^2 (\hat{\eta}^2) \\
&= \left[ \begin{matrix}
\frac{s}{v} \sigma^2 (\hat{\eta}^2) & \left( \frac{1}{v} + \frac{1}{s} \right) \sigma^2 (\hat{\eta}^2) & \frac{4}{s} \sigma^2 (\hat{\eta}^2)
\end{matrix} \right]
\end{align}

\section{Proof of Theorem \ref{th:ate_part_trig}}
Estimated product trigger intensity is represented by a column vector $\mathbf{r}^{\prime}$ of size $N \times 1$. The estimation error for product trigger intensity is $\boldsymbol{\epsilon}$, which is a column vector of size $N \times 1$. Hence, $\mathbf{r}^{\prime} = \mathbf{r} + \boldsymbol{\epsilon}$. We can write
\begin{align}
&\mathbf{E} = \left[ \begin{matrix}
 \mathbf{0} & \mathbf{e} & \mathbf{e} \odot \mathbf{t}
 \end{matrix} \right]\\
 &\tilde{\mathbf{X}} = \mathbf{X} + \mathbf{E}
\end{align}
Assume, $E[r_i^{\prime}] = m^{\prime}$, $E[(r_i^{\prime})^2] = s^{\prime}$, $\sigma^2(r_i^{\prime}) = v^{\prime}$, $E[\epsilon_i] = e$, $E[\epsilon_i^2] = q$, and $E[\epsilon_i r_i^{\prime}] = k$. The estimated parameters are $\hat{\beta^{\prime}}$.

\subsection{Proof of Theorem \ref{th:ate_part_trig}: part a}
\begin{align}
&\tilde{\mathbf{X}} = \left[ \begin{matrix}
\mathbf{1} & \mathbf{r^{\prime}} & \mathbf{r^{\prime}} \odot \mathbf{t}
\end{matrix} \right] \\
&\frac{1}{N}\tilde{\mathbf{X}^{T}} \tilde{\mathbf{X}} = \left[ \begin{matrix}
1 & E[r_i^{\prime}] & E[\mathcal{T}_ir_i^{\prime}]\\
E[r_i^{\prime}] & E[{r_i^{\prime}}^2] & E[\mathcal{T}_i{r_i^{\prime}}^2]\\
E[\mathcal{T}_i r_i^{\prime}] & E[\mathcal{T}_i {r_i^{\prime}}^2] & E[\mathcal{T}_i^2 {r_i^{\prime}}^2]
\end{matrix} \right] = \left[ \begin{matrix}
1 & m^{\prime} & \frac{m^{\prime}}{2}\\
m^{\prime} & s^{\prime} & \frac{s^{\prime}}{2}\\
\frac{m^{\prime}}{2} & \frac{s^{\prime}}{2} & \frac{s^{\prime}}{2}
\end{matrix} \right]\\
\label{eqn:est_trig_rate_1}
&\left( \frac{1}{N}\tilde{\mathbf{X}}^{T} \tilde{\mathbf{X}} \right)^{-1} = \left[ \begin{matrix}
\frac{s^{\prime}}{v^{\prime}} & -\frac{m^{\prime}}{v^{\prime}} & 0\\
-\frac{m^{\prime}}{v^{\prime}} & \frac{1}{v^{\prime}} + \frac{1}{s^{\prime}} & -\frac{2}{s^{\prime}}\\
0 & -\frac{2}{s^{\prime}} & \frac{4}{s^{\prime}}
\end{matrix} \right]\\
\label{eqn:beta_hat_prime}
&\boldsymbol{\hat{\beta^{\prime}}} = \left( \frac{1}{N} \tilde{\mathbf{X}}^{T} \tilde{\mathbf{X}} \right)^{-1} \frac{1}{N} \tilde{\mathbf{X}}^{T} \mathbf{y}\\
&= \left[ \begin{matrix}
\frac{s^{\prime}}{v^{\prime}} E[y_i] - \frac{m^{\prime}}{v^{\prime}}E[r_i^{\prime}y_i]\\
-\frac{m^{\prime}}{v^{\prime}}E[y_i] + \frac{1}{v^{\prime}} E[r_i^{\prime}y_i] - \frac{1}{2s^{\prime}} \left( E[r_i^{\prime}y_i | \mathcal{T}_i = 1] - E[r_i^{\prime}y_i | \mathcal{T}_i = 0] \right)\\
\frac{1}{s^{\prime}} \left( E[r_i^{\prime}y_i | \mathcal{T}_i = 1] - E[r_i^{\prime}y_i | \mathcal{T}_i = 0] \right)
\end{matrix} \right]
\end{align}
From Eqn \ref{eqn:beta_hat_prime}, we can say that
\begin{align}
\hat{\beta_2^{\prime}} = \frac{1}{E[(r_i^{\prime})^2]} \left[ E[y_i r_i^{\prime} | \mathcal{T}_i = 1] - E[y_i r_i^{\prime} | \mathcal{T}_i = 0] \right]
\end{align}

\subsection{Proof of Theorem \ref{th:ate_part_trig}: part b}
\begin{align}
\label{eqn:est_trig_rate_2}
&\frac{1}{N} \tilde{\mathbf{X}}^{T} \mathbf{E} = \left[ \begin{matrix}
0 & E[\epsilon_i] & E[\mathcal{T}_i \epsilon_i]\\
0 & E[r_i^{\prime} \epsilon_i] & E[\mathcal{T}_i r_i^{\prime} \epsilon_i]\\
0 & E[\mathcal{T}_i r_i^{\prime} \epsilon_i] & E[\mathcal{T}_i r_i^{\prime} \epsilon_i]
\end{matrix} \right] = \left[ \begin{matrix}
0 & e & \frac{e}{2}\\
0 & k & \frac{k}{2}\\
0 & \frac{k}{2} & \frac{k}{2}
\end{matrix} \right]\\
\label{eqn:true_param_w_est_noise}
&\left( \frac{1}{N} \tilde{\mathbf{X}}^{T} \tilde{\mathbf{X}} \right)^{-1} \frac{1}{N} \tilde{\mathbf{X}}^{T} \mathbf{E} = \left[ \begin{matrix}
0 & \frac{s^{\prime}e - m^{\prime}k}{v^{\prime}} & \frac{s^{\prime}e - m^{\prime}k}{2v^{\prime}}\\
0 & \frac{k - e m^{\prime}}{v^{\prime}} & \frac{k - e m^{\prime}}{2v^{\prime}} - \frac{k}{s^{\prime}}\\
0 & 0 & \frac{k}{s^{\prime}}
\end{matrix} \right]
\end{align}

The true parameters are $\boldsymbol{\beta}$ and the noise is $\boldsymbol{\eta}$. 
\begin{align}
\label{eqn:data_gen_proc_mat}
\mathbf{y} = \mathbf{X} \boldsymbol{\beta} + \boldsymbol{\eta}
\end{align}
\begin{align}
\label{eqn:beta_hat_prime_mat}
\hat{\boldsymbol{\beta}^{\prime}} &= \left( \tilde{\mathbf{X}}^{T} \tilde{\mathbf{X}} \right)^{-1} \tilde{\mathbf{X}}^{T} \mathbf{y}\\
&= \left( \tilde{\mathbf{X}}^{T} \tilde{\mathbf{X}} \right)^{-1} \tilde{\mathbf{X}}^{T} \left( \mathbf{X} \boldsymbol{\beta} + \boldsymbol{\eta} \right)\\
& = \left( \tilde{\mathbf{X}}^{T} \tilde{\mathbf{X}} \right)^{-1} \tilde{\mathbf{X}}^{T} \left( (\tilde{\mathbf{X}} - \mathbf{E}) \boldsymbol{\beta} + \boldsymbol{\eta} \right)\\
&= \boldsymbol{\beta} - \left( \tilde{\mathbf{X}}^{T} \tilde{\mathbf{X}} \right)^{-1} \tilde{\mathbf{X}}^{T} \mathbf{E} \boldsymbol{\beta} + \left( \tilde{\mathbf{X}}^{T} \tilde{\mathbf{X}} \right)^{-1} \tilde{\mathbf{X}}^{T} \boldsymbol{\eta}
\end{align}

We can use Eq~\ref{eqn:true_param_w_est_noise} to re-write Eq \ref{eqn:beta_hat_prime_mat} as
\begin{align}
\label{eqn:beta_hat_prime_mat_bias}
    \hat{\boldsymbol{\beta}^{\prime}} &= \left[ \begin{matrix}
1 & -\frac{s^{\prime}e - m^{\prime}k}{v^{\prime}} & -\frac{s^{\prime}e - m^{\prime}k}{2v^{\prime}}\\
0 & 1 - \frac{k - e m^{\prime}}{v^{\prime}} & -\frac{k - e m^{\prime}}{2v^{\prime}} + \frac{k}{s^{\prime}}\\
0 & 0 & 1 - \frac{k}{s^{\prime}} 
\end{matrix} \right] \boldsymbol{\beta} + \left( \tilde{\mathbf{X}}^{T} \tilde{\mathbf{X}} \right)^{-1} \tilde{\mathbf{X}}^{T} \boldsymbol{\eta}
\end{align}
Hence, $E[\hat{\beta_2^{\prime}}] = \beta_2 \left[ 1 - \frac{ E[\epsilon_i r_i^{\prime}]}{E[(r_i^{\prime})^2]} \right]$

\subsection{Proof of Theorem \ref{th:ate_part_trig}: part c}
The residuals for OLS using estimated trigger intensity is
\begin{align}
\hat{\boldsymbol{\eta^{\prime}}} &= \mathbf{y} - \tilde{\mathbf{X}}\hat{\boldsymbol{\beta}}^{\prime}\\
&= \boldsymbol{\eta} + \mathbf{X} \boldsymbol{\beta} - \tilde{\mathbf{X}}\hat{\boldsymbol{\beta}}^{\prime} \; \text{from Eq \ref{eqn:data_gen_proc_mat}}\\
&= \boldsymbol{\eta} + \mathbf{X} \boldsymbol{\beta} - \tilde{\mathbf{X}} \left( \boldsymbol{\beta} - \left( \tilde{\mathbf{X}}^{T} \tilde{\mathbf{X}} \right)^{-1} \tilde{\mathbf{X}}^{T} \mathbf{E} \boldsymbol{\beta} + \left( \tilde{\mathbf{X}}^{T} \tilde{\mathbf{X}} \right)^{-1} \tilde{\mathbf{X}}^{T} \boldsymbol{\eta} \right)\\
&= \boldsymbol{\eta} + \left( \tilde{\mathbf{X}} - \mathbf{E} \right) \boldsymbol{\beta} - \tilde{\mathbf{X}} \left( \boldsymbol{\beta} - \left( \tilde{\mathbf{X}}^{T} \tilde{\mathbf{X}} \right)^{-1} \tilde{\mathbf{X}}^{T} \mathbf{E} \boldsymbol{\beta} + \left( \tilde{\mathbf{X}}^{T} \tilde{\mathbf{X}} \right)^{-1} \tilde{\mathbf{X}}^{T} \boldsymbol{\eta} \right)\\
&= \boldsymbol{\eta} - \tilde{\mathbf{X}} \left( \tilde{\mathbf{X}}^{T} \tilde{\mathbf{X}} \right)^{-1} \tilde{\mathbf{X}}^{T} \boldsymbol{\eta} - \mathbf{E} \boldsymbol{\beta} + \tilde{\mathbf{X}} \left( \tilde{\mathbf{X}}^{T} \tilde{\mathbf{X}} \right)^{-1} \tilde{\mathbf{X}}^{T} \mathbf{E} \boldsymbol{\beta} \\
\label{eqn:noise_with_expand}
&= \left( \mathbf{I}_{N \times N} -  \tilde{\mathbf{X}} \left( \tilde{\mathbf{X}}^{T} \tilde{\mathbf{X}} \right)^{-1} \tilde{\mathbf{X}}^{T} \right) \boldsymbol{\eta} - \left( \mathbf{I}_{N \times N} -  \tilde{\mathbf{X}} \left( \tilde{\mathbf{X}}^{T} \tilde{\mathbf{X}} \right)^{-1} \tilde{\mathbf{X}}^{T} \right) \mathbf{E} \boldsymbol{\beta}
\end{align}
Suppose $\mathbf{A} = \left( \mathbf{I}_{N \times N} -  \tilde{\mathbf{X}} \mathbf{M} \right)$. $\mathbf{A}$ is a special matrix. It satisfies the following properties.
\begin{itemize}
    \item $\mathbf{A}$ is a symmetric matrix.
    \item $\mathbf{A}$ is idempotent i.e., $\mathbf{A}^T \mathbf{A} = \mathbf{A}$.
    \item $\mathbf{A}$ has $N - 3$ non-zero eigenvalues and 3 zero eigenvalues. All of its non-zero eigenvalues are 1. Hence $\text{tr} \left( \mathbf{A} \right) = N - 3$.
\end{itemize}
We can write Eqn \ref{eqn:noise_with_expand} as
\begin{align}
\hat{\boldsymbol{\eta^{\prime}}} = \mathbf{A} \left( \boldsymbol{\eta} -  \mathbf{E} \boldsymbol{\beta} \right)
\end{align}
Now, the expected value of the residuals is zero, so 
\begin{align}
    \sigma^2(\hat{\eta^{\prime}}) = \frac{1}{N} E[\hat{\boldsymbol{\eta^{\prime}}}^{T} \hat{\boldsymbol{\eta^{\prime}}}]
\end{align}

\begin{align}
E[\hat{\boldsymbol{\eta^{\prime}}}^{T} \hat{\boldsymbol{\eta^{\prime}}}] &= E[\left( \mathbf{A} \left( \boldsymbol{\eta} -  \mathbf{E} \boldsymbol{\beta} \right) \right)^T \left( \mathbf{A} \left( \boldsymbol{\eta} -  \mathbf{E} \boldsymbol{\beta} \right) \right)]\\
&= E[\boldsymbol{\eta}^T \mathbf{A}^T \mathbf{A} \boldsymbol{\eta}] + E[\boldsymbol{\beta}^T \mathbf{E}^T \mathbf{A}^T \mathbf{A} \mathbf{E} \boldsymbol{\beta}]\\
\label{eqn:residual_exapnsion}
&= E[\boldsymbol{\eta}^T \mathbf{A} \boldsymbol{\eta}] + \boldsymbol{\beta}^T E[ \mathbf{E}^T \mathbf{A} \mathbf{E} ] \boldsymbol{\beta}
\end{align}

We can show that
\begin{align}
\label{eqn:e_a_e^t}
    E[ \mathbf{E}^T \mathbf{A} \mathbf{E} ] &= \text{tr}(A) \left[ \begin{matrix}
0 & 0 & 0 \\
0 & q & \frac{q}{2} \\
0 & \frac{q}{2} & \frac{q}{2}
\end{matrix} \right]
\end{align}

Based on Eqn \ref{eqn:e_a_e^t}, we can write Eqn \ref{eqn:residual_exapnsion} as
\begin{align}
    E[\hat{\boldsymbol{\eta^{\prime}}}^{T} \hat{\boldsymbol{\eta^{\prime}}}] = \sigma^2(\eta) \text{tr}(A) + \left[ \left( \beta_1 + \frac{\beta_2}{2} \right)^2 + \frac{\beta_2^2}{4} \right] q \text{tr}(A)
\end{align}

When $N$ is very large, we can approximate $\sigma^2(\hat{\eta^{\prime}})$ as
\begin{align}
    \sigma^2(\hat{\eta^{\prime}}) &\approx \sigma^2(\eta) + \left[ \left( \beta_1 + \frac{\beta_2}{2} \right)^2 + \frac{\beta_2^2}{4} \right] q
\end{align}

\subsection{Proof of Theorem \ref{th:ate_part_trig}: part d}
Suppose, $\sigma^2 (\hat{\eta^{\prime}})$ is the variance of the residuals, then
\begin{align}
&\sigma^2(\boldsymbol{\hat{\beta^{\prime}}}) = \text{diag} \left( \frac{1}{N} \tilde{\mathbf{X}}^{T} \tilde{\mathbf{X}} \right)^{-1} \sigma^2 (\hat{\eta^{\prime}})\\
&= \left[ \begin{matrix}
\frac{s^{\prime}}{v^{\prime}} \sigma^2 (\hat{\eta^{\prime}}) & \left( \frac{1}{v^{\prime}} + \frac{1}{s^{\prime}} \right) \sigma^2 (\hat{\eta^{\prime}}) & \frac{4}{s^{\prime}} \sigma^2 (\hat{\eta^{\prime}})
\end{matrix} \right]
\end{align}
Hence, $\sigma^2 \left( \hat{\beta_2^{\prime}} \right) = \frac{4}{E[r_i^{{\prime}^2}]}  \sigma^2 (\hat{\eta^{\prime}})$

\section{Proof of Theorem \ref{th:ub_ate_part_trig}}
Here, $E[\epsilon_i] = 0$ and $E[\epsilon_i r_i] = 0$. Thus, $E[r_i^{\prime}] = E[r_i]$ and $E[r_i^{{\prime}^2}] = E[r_i^2] + E[\epsilon_i^2]$
\subsection{Proof of Theorem \ref{th:ub_ate_part_trig}: part a}
\begin{proof}
    We can re-write the bias in estimated $\hat{\rho^{\prime}}$.
    \begin{align}
        \rho - E[\hat{\rho^{\prime}}] &= \beta_2 \left(\frac{ E[\epsilon_i r_i^{\prime}]E[r_i^{\prime}]}{E[(r_i^{\prime})^2]} - E[\epsilon_i] \right)\\
        &= \beta_2 \frac{ E[\epsilon_i r_i^{\prime}]E[r_i^{\prime}]}{E[(r_i^{\prime})^2]}\\
        &= \beta_2 \frac{E[\epsilon_i^2] E[r_i]}{E[r_i^2] + E[\epsilon_i^2]}\\
        &= \beta_2 \frac{(E[r_i] - E[r_i^2])E[r_i]}{(m - 1) E[r_i^2] + E[r_i]} \text{ from Lemma \ref{lemma:ub_ate_part_trig_1}}\\
        &< \beta_2 \frac{E[r_i]^2}{(m - 1) E[r_i^2] + E[r_i]}\\
        &< \beta_2 \frac{E[r_i]^2}{(m - 1) E[r_i^2]}\\
        &< \frac{\beta_2}{m - 1}
    \end{align}
\end{proof}

\subsection{Proof of Theorem \ref{th:ub_ate_part_trig}: part b}
\begin{align}
&\sigma^2 \left( \hat{\rho}^{\prime} \right) = \frac{4 E[r_i^{\prime}]^2}{E[r_i^{{\prime}^2}]}  \sigma^2 (\hat{\eta^{\prime}})\\
&= \frac{4 E[r_i]^2}{E[r_i^2] + E[\epsilon_i^2]} \sigma^2 (\hat{\eta^{\prime}})\\
&= \frac{4 E[r_i]^2}{E[r_i^2] + E[\epsilon_i^2]} \left[ \sigma^2(\eta) + \left[ \left( \beta_1 + \frac{\beta_2}{2} \right)^2 + \frac{\beta_2^2}{4} \right] E[\epsilon_i^2] \right]\\
&< \frac{4 E[r_i]^2}{E[r_i^2]} \left[ \sigma^2(\eta) + \left[ \left( \beta_1 + \frac{\beta_2}{2} \right)^2 + \frac{\beta_2^2}{4} \right] E[\epsilon_i^2] \right]\\
&= \sigma^2 \left( \hat{\rho}\right) + \frac{4 E[r_i]^2}{E[r_i^2]} \left[ \left( \beta_1 + \frac{\beta_2}{2} \right)^2 + \frac{\beta_2^2}{4} \right] E[\epsilon_i^2],  \sigma^2(\eta) = \sigma^2(\hat{\eta})\\
&< \sigma^2 \left( \hat{\rho}\right) + 4 E[\epsilon_i^2] \left[ \left( \beta_1 + \frac{\beta_2}{2} \right)^2 + \frac{\beta_2^2}{4} \right], \frac{ E[r_i]^2}{E[r_i^2]} < 1 \notag\\
&\text{, when } \sigma^2 (r_i) > 0 \\
&= \sigma^2 \left( \hat{\rho}\right) + 4 \frac{E[r_i] - E[r_i^2]}{m} \left[ \left( \beta_1 + \frac{\beta_2}{2} \right)^2 + \frac{\beta_2^2}{4} \right] \notag\\
& \text{ from Lemma \ref{lemma:ub_ate_part_trig_1}}\\
&< \sigma^2 \left( \hat{\rho}\right) + \frac{1}{m} \left[ \left( \beta_1 + \frac{\beta_2}{2} \right)^2 + \frac{\beta_2^2}{4} \right], E[r_i] - E[r_i^2] \le \frac{1}{4}\\
&\sigma^2 \left( \hat{\rho}^{\prime} \right) - \sigma^2 \left( \hat{\rho}\right) < \frac{1}{m} \left[ \left( \beta_1 + \frac{\beta_2}{2} \right)^2 + \frac{\beta_2^2}{4} \right]
\end{align}
\end{document}